\documentclass[a4paper,UKenglish,cleveref, autoref, thm-restate]{lipics-v2021}

\title{Security Analysis of Filecoin's Expected Consensus in the Byzantine vs Honest Model} 
\titlerunning{EC Security Analysis}

\author{Xuechao Wang\footnote{Corresponding author, part of work was done at Protocol Labs.}}{Thrust of Financial Technology, HKUST(GZ), China}{xuechaowang@hkust-gz.edu.cn}{https://orcid.org/0000-0001-6918-2699}{}

\author{Sarah Azouvi}{Protocol Labs, USA}{s.azouvi@gmail.com}{https://orcid.org/0000-0002-7133-1937}{}

\author{Marko Vukolić}{Protocol Labs, USA}{marko.vukolic@protocol.ai}{https://orcid.org/0000-0002-9898-5383}{}

\authorrunning{X. Wang, S. Azouvi and M. Vukolić} 

\Copyright{Xuechao Wang, Sarah Azouvi and Marko Vukolić} 



\category{} 

\relatedversion{} 

\acknowledgements{The authors would like to thank Guy Goren for his suggestion of Consistent Broadcast as a mitigation to the $n$-split attack.}

\nolinenumbers 

\EventEditors{Joseph Bonneau and S. Matthew Weinberg}
\EventNoEds{2}
\EventLongTitle{5th Conference on Advances in Financial Technologies (AFT 2023)}
\EventShortTitle{AFT 2023}
\EventAcronym{AFT}
\EventYear{2023}
\EventDate{October 23-25, 2023}
\EventLocation{Princeton, NJ, USA}
\EventLogo{}
\SeriesVolume{282}
\ArticleNo{5}
 
\usepackage[utf8]{inputenc}
\usepackage{graphicx}
\usepackage[T1]{fontenc}
\usepackage{url}
\usepackage{ifthen}
\usepackage{enumitem,wrapfig}
\usepackage{amsmath}
\usepackage{amsthm}
\usepackage{textcomp}
\usepackage{siunitx}
\usepackage{color}
\usepackage[countmax]{subfloat}
\usepackage{cases}
\usepackage[font={small}]{caption}
\usepackage{subcaption}
\usepackage{tcolorbox}
\usepackage{fix-cm}
\usepackage{bbm}
\usepackage{hyperref}
\usepackage{algorithm}
\usepackage[noend]{algpseudocode}
\usepackage{algorithm}

\usepackage{xcolor}
\definecolor{azure}{rgb}{0.54, 0.17, 0.89}

\makeatletter
\def\thm@space@setup{\thm@preskip=2pt
\thm@postskip=2pt \itshape}
\makeatother
\newtheoremstyle{newstyle}
{} 
{} 
{\mdseries} 
{} 
{\bfseries} 
{.} 
{ } 
{} 

\theoremstyle{newstyle}

\theoremstyle{definition}

\theoremstyle{remark}

\newcommand{\C}{{\mathcal{C}}}

\newcommand{\target}{\textsf{target}}
\newcommand{\drand}{\textsf{drand}}
\newcommand{\vrfproof}{\textsf{VRF.Proof}}
\newcommand{\vrfverify}{\textsf{VRF.Verify}}
\newcommand{\pk}{\textsf{pk}}
\newcommand{\sk}{\textsf{sk}}
\newcommand{\seed}{\textsf{seed}}
\newcommand{\tipset}{\mathcal{T}}
\newcommand{\chain}{\mathcal{C}}
\newcommand{\block}{\mathcal{B}}
\newcommand{\parent}{\mathsf{parent}}
\newcommand{\epoch}{\mathsf{epoch}}

\newif\iffcsubmission\fcsubmissionfalse

\begin{document}
\def\displayrev{0} 
\maketitle

\begin{abstract}
Filecoin is the largest storage-based open-source blockchain, both by storage capacity (>11EiB) and market capitalization. This paper provides the first formal security analysis of Filecoin's consensus (ordering) protocol, Expected Consensus (EC). 
Specifically,
 we show that EC is secure against an arbitrary adversary that controls a fraction $\beta$ of the total storage for $\beta m< 1- e^{-(1-\beta)m}$, where $m$ is a parameter that corresponds to the expected number of blocks per round, currently $m=5$ in Filecoin.
 We then present an attack, the $n$-split attack, where an adversary splits the honest miners between multiple chains, and
show that it is successful for $\beta m \ge 1- e^{-(1-\beta)m}$, thus proving that $\beta m= 1- e^{-(1-\beta)m}$ is the tight security threshold of EC. This corresponds roughly to an adversary with $20\%$ of the total storage pledged to the chain.
 Finally, we propose two improvements to EC security that would increase this threshold. One of these two fixes is being implemented as a Filecoin Improvement Proposal (FIP).

\end{abstract}

\section{Introduction}
Filecoin is the largest storage-based blockchain in terms of both market cap~\cite{coinmarketcap} and total raw-byte storage capacity (>11EiB)~\cite{filfox}.
In Filecoin, miners, called Storage Providers (SPs),  gain the right to participate in the consensus protocol and to create blocks by pledging  storage capacity to the chain.\footnote{Filecoin further incentivizes the storage of ``useful'' data, where SPs have the additional opportunity to boost their raw-byte storage power, offering deals to verified clients, yielding \emph{quality adjusted} power \cite{miningguide}.} They are in return compensated with a financial reward in the form of newly minted FIL, the cryptocurrency underlying Filecoin, whenever their blocks are included on-chain, where probability of a miner minting new block corresponds to their storage power. The Filecoin consensus mechanism Storage Power Consensus (SPC) consists mainly
of two components: first, a \emph{Sybil-resistance mechanism} that keeps an accurate map of the storage pledged by each storage provider; and second, a consensus protocol that can be run by any set of weighted participants and outputs an ordered list of transactions. 

In this paper, we ignore the mechanisms that keep the mapping between miners and their respective storage accurate (i.e., the Sybil-resistance mechanism and the quality adjusted power policy) and focus on the sub-protocol run by the weighted miners to produce an ordered list of transactions. This sub-protocol is called \textit{Expected Consensus} (EC) and weighs each participant according to their storage power.  We assume that the weighted list of miners is accurately maintained and given as an input to EC.
EC is a longest-chain style protocol~\cite{filecoin} (or, more accurately, a heaviest-chain protocol). At a high level, it operates by running a leader election at every time slot in which, on expectation, $m$ participants may be eligible to submit a block, where $m$ is a parameter currently equal to 5.
Each participant is elected with probability proportional to their weight.
Multiple valid blocks submitted in a given round may form a \emph{tipset}, which is a set of blocks sharing the same height (i.e., round number) and parent tipset. In EC, the blockchain is a chain of tipsets (i.e., a directed acyclic graph [DAG] of blocks) rather than a chain of blocks.
For example, in Figure~\ref{fig:tipset-1} blocks \{A,B,C\}, \{D,E\} and \{F,G,H\} each form a different tipset; and in Figure~\ref{fig:tipset-2} blocks \{A,B,C\}, \{D,E\} and \{F\} each form a different tipset. Every block in a tipset adds weight to its chain of tipsets, while the fork choice rule is to choose the heaviest tipset.
EC works in a very similar fashion as longest-chain protocols like Bitcoin do, but it uses tipsets instead of blocks. EC's security has, until now, only been argued informally, as with Bitcoin in its early days. Intuitively, tipsets make it harder for an adversary with less storage to form a competing chain of tipsets with more blocks than the main chain, as miners can create a number of blocks proportional to their storage power. 
This is analogous to Nakamoto's private attack on the longest-chain protocol~\cite{nakamoto2008peer}.
Specifically, assuming that two competing chains of tipsets are growing, with different amount of storage power producing the two, since in EC more than one block can be appended to a chain at each round, the difference between the number of blocks created on each chain will grow roughly $m$ times faster than in the case without tipset (i.e., where each chain can grow by at most one block at each round). However, this intuitive security justification only applies when examining the private attack, a specific instance, not a general adversary. The lessons learned from the balance attack~\cite{natoli2016balance} on {\sf GHOST}~\cite{sompolinsky2015secure} highlight that a comprehensive analysis encompassing all potential attacks is crucial for assuring the security of a blockchain protocol. This comprehensive evaluation is the main focus of this paper.

In this work, we conduct a formal security analysis of EC and prove that EC is secure against any adversary that owns a fraction $\beta$ of the total storage power for $\beta m< 1- e^{-(1-\beta)m}$ (Section~\ref{sec:proof}). To achieve this, we carefully extend the proof technique developed in~\cite{dembo2020everything} to EC: the key step is to identify the sufficient condition for a block to stay in the chain forever, regardless of the complex DAG structure in EC. Indeed, the incorporation of tipsets introduces substantial complexities to the problem. For instance, in the longest-chain protocol, the chain growth is an independent and identically distributed random variable in each round. Conversely, within EC, the chain growth becomes dependent on the entire history of the blockchain, given that it depends on the structure of the DAG. This increased dependency adds layers of intricacy to the security analysis.
Following similar literature~\cite{dembo2020everything,backbone, kiayias2017ouroboros}, we consider a rather strong adversary, which we specify in Section~\ref{sec:model}, that has ``full control'' over both the network and the tie breaking rule.
We then propose an attack, the $n$-split attack, in Section~\ref{sec:attack} in which an adversary with power $\beta$ such that $\beta m\ge 1- e^{-(1-\beta)m}$ can break the security of EC, effectively proving that the security threshold $\beta m= 1- e^{-(1-\beta)m}$ is tight for EC. 
With current parameters, this means that EC is secure against an adversary that holds roughly 20\% of the total storage power.
In our attack, an elected leader, controlled by the adversary, equivocates by sending different blocks to different miners at each round with the aim to split the honest miners into different chains and thus reducing the weight of each tipsets' chain. While the honest participants are split and all mine on potentially many different forks, the adversary can construct a \emph{private} tipsets' chain on the side, i.e., a chain that does not include any block mined by an honest miner (called for simplicity \emph{honest block}) and that will not be broadcast to any honest miner until the end of the attack.
Finally, in Section~\ref{sec:mitigation}, we propose two countermeasures against this attack aimed at augmenting the security threshold of Filecoin. The first one entails eliminating the concept of tipsets, substituting EC with the longest-chain protocol~\cite{kiayias2017ouroboros,david2018ouroboros,badertscher2018ouroboros}, as our observations suggest that lowering the value of $m$ can enhance the security threshold. The second approach involves the adoption of a consistent or reliable broadcast protocol~\cite{bracha1985asynchronous,guerraoui2019scalable} to prohibit the adversary from equivocating. Our second countermeasure is currently in the process of being implemented as a Filecoin Improvement Proposal~\cite{fip}.

\begin{figure}
\centering
\begin{subfigure}[t]{.4\textwidth}
  \centering
  \includegraphics[width=\linewidth]{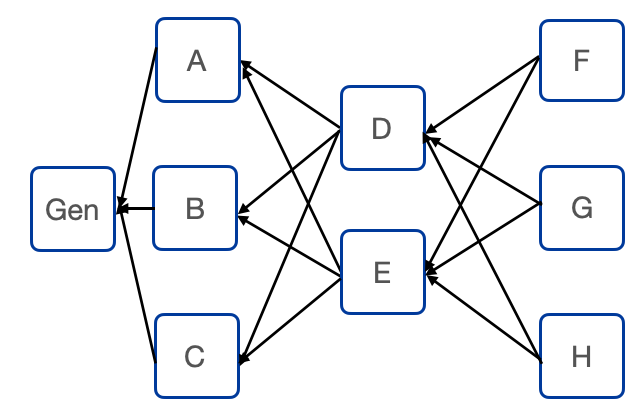}
  \caption{Example of a tipset chain: \{A,B,C\} is the first tipset after Genesis, and the parent tipset of \{D,E\}, which is itself the parent tipset of tipset \{F,G,H\}.}
  \label{fig:tipset-1}
\end{subfigure}%
\hspace{1em}%
\begin{subfigure}[t]{.4\textwidth}
  \centering
    \includegraphics[width=0.6\linewidth]{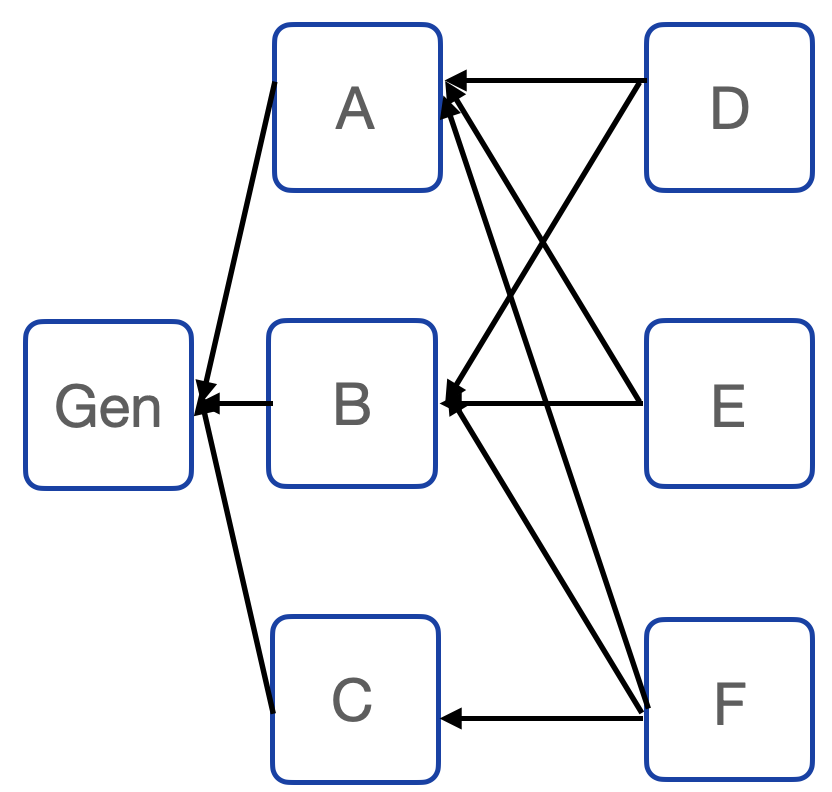}
    \caption{Example of two tipsets with equal weight. Tipset \{D,E\} and tipset \{F\} both have a weight of 5. (roughly, the weight is calculated as the sum of the blocks in all the tipsets, see Section~\ref{sec:fcr} for details).}
  \label{fig:tipset-2}
\end{subfigure}
\caption{Tipset structure in EC.}
\label{fig:tipsets}
\end{figure}

\subsubsection*{Related Work}
This work is directly inspired by the line of work formally analyzing  longest-chain protocols either in the proof-of-work case~\cite{backbone,ren2019analysis,kiffer2018better} or proof-of-stake~\cite{bagaria2019proof,kiayias2017ouroboros,dembo2020everything,pass2017sleepy}.
We adapt the technique in~\cite{dembo2020everything} to account for tipsets instead of blocks (see Lemma~\ref{lem:converge}).
The main difference between using a chain of tipsets and a chain of blocks is that in the tipset case, the number of blocks in the chain at each round can increase by any finite integer value and also depends on the structure of the DAG. By contrast, in the chain-of-blocks case, the number of blocks of any honest chain increases by zero or one at each round. This makes the tipset analysis more complex.

Similar attacks to the one we describe in Section~\ref{sec:attack} were proposed by Bagaria et al.~\cite{bagaria2019proof} in the context of proof-of-stake and by Natoli and Gramoli~\cite{natoli2016balance} in the proof-of-work context. In both these works, the attacks are described on a DAG-based blockchain, where each new block can include any previous block as its parent. We adapt it to the tipset case, which is slightly different: in the case of tipsets, a block may have multiple parents, but only blocks that have themselves the same set of parents can be referenced by the child block.
The idea behind these attacks is to have an adversary publish its blocks in a timely manner on different forks to ensure that honest miners keep on extending two or even more chains instead of one, effectively spreading their \emph{power} (be it stake, computation or storage) on different chains.

\section{Model}
\label{sec:model}
In this section we present our model and assumptions.
\subsection{Participants}
Filecoin requires participants to pledge storage capacity to the chain to be added to the list of participants.
Following the work in~\cite{backbone,kiayias2017ouroboros} we consider a flat model, meaning that each participant accounts for one unit of storage. This could easily be extended to a non-flat model by considering that a participant holding $x$ units of storage controls $x$ ``flat'' participants.
We consider a static model wherein the set of participants is fixed during the execution of the protocol. 
We assume that each participant $i$ possesses a key pair $(\textsf{sk}_i,\textsf{pk}_i)$
and that every participant is aware of the other participants and their respective public keys.

We consider a static adversary that corrupts a fraction $\beta$ of the participants at the beginning of the execution of the protocol. In order to defend against an adaptive adversary who can corrupt honest nodes on the fly (i.e., dynamically, at the start of each round), one can either use key evolving signature schemes~\cite{david2018ouroboros} or checkpointing~\cite{azouvi2022pikachu}. However, in order to keep the problem simple, we do not consider an adaptive adversary in this paper.

\subsection{Network assumptions}\label{sec:model-network}
We consider the lock-step synchronous network model adopted in~\cite{backbone,kiayias2017ouroboros}. Time is divided into synchronized rounds, each indexed with an integer in $\mathbb{N}$. Following Filecoin's terminology~\cite{filecoin}, we refer to each time slot as an \emph{epoch}. Each epoch has a fixed duration (currently 30 seconds in Filecoin). To abstract the underlying peer-to-peer gossip network in Filecoin, we simply assume that all messages sent by honest nodes are broadcast to all nodes and that all honest nodes re-broadcast any message they deliver. All network messages are delivered by the adversary, and we allow the adversary to selectively delay messages sent by honest nodes, with the following restrictions: (i) the messages sent in an epoch must be delivered by the end of the current epoch; and (ii) the adversary cannot forge or alter any message sent by an honest node. The adversary does not suffer any network delay. Note that the adversary can selectively send its message only to a subset of honest nodes. However, due to the re-broadcast mechanism, all honest nodes will receive the message by the end of the next epoch. 

The non-lock-step synchronous model, also known as the $\Delta$-synchronous model, is also frequently employed in blockchain security analysis~\cite{badertscher2018ouroboros,david2018ouroboros,dembo2020everything,pass2017analysis,ren2019analysis}. This model ensures messages sent by honest nodes are delivered within a sliding window of $\Delta$ epochs. While this model might be more suitable for proof-of-work blockchains, where miners persistently mine and broadcast blocks, its application becomes less pertinent for PoS blockchains. In the latter, honest nodes primarily remain dormant, awaiting the epoch boundaries to send messages. Given the efficiency of today's network infrastructure, a 30-second window adopted by EC is quite conservative. Consequently, we find little justification to incorporate the $\Delta$-synchronous model in our EC analysis.

\subsection{Randomness}

\noindent{\bf Random beacon.}
A random beacon~\cite{rabin1983transaction} is a system that emits a random number at regular intervals. 
EC relies on drand~\cite{drand}, a decentralized random beacon, to provide miners a different random number at each epoch. 
This service is run by a set of 16 independent institutions that run a multi-party protocol to output, at regular intervals, a fresh random number.
We assume that this random number is unbiasable (i.e., truly random) and unpredictable before the beginning of the epoch. We also assume that each miner in Filecoin has the same view of each drand output, i.e., that drand is secure. We denote $\drand_i$ the random beacon emitted by drand and used by Filecoin miners at epoch $i$.\\

\noindent{\bf Verifiable Random Function.}
A Verifiable Random Function (VRF) \cite{micali1999verifiable} is a function that outputs a random number in a verifiable way, i.e., everyone can verify that the output is indeed random and was generated correctly.
A VRF is composed of two polynomial-time algorithms: $\vrfproof$ and $\vrfverify$ (we omit the key generation).
$\vrfproof$ takes as inputs a seed $\textsf{seed}$ and a secret key $\textsf{sk}$ and outputs a tuple $(y=G_\sk(\seed),p=\pi_\sk(\seed))$ where $y$ is a random number and $p$ is a proof that can be used to verify the correctness of $y$.
$\vrfverify$ takes as input a tuple $(\seed,y,p)$ and a public key $\textsf{pk}$ and uses $p$ to verify that $y=G_\sk(\seed)$, in which case it outputs 1; otherwise, it outputs 0.
A VRF is correct if:

\begin{enumerate}
    \item if $(y,p) = \vrfproof_\sk(\seed)$, then $\vrfverify_\pk(\seed,y,p)=1$;
    \item for all $(\sk,\seed)$, there is a unique $y$ s.t. $\vrfverify_\pk(\seed,y,\pi_\sk(\seed))=1$;
    \item $G_\sk(\seed)$ is computationally indistinguishable from a random number for any probabilistic polynomial-time adversary.
\end{enumerate}

Throughout the rest of this paper, we assume the existence of a correct VRF.

\section{Filecoin's Expected Consensus (EC)}
Filecoin's consensus protocol, EC,
consists of three main components: a leader election sub-protocol, a mining algorithm and a fork choice rule. Briefly, at the beginning of each epoch, participants will check their eligibility to produce a block by running the leader election. 
If they are elected, they use the fork choice rule to select a tipset and include it as their \emph{parent} before broadcasting their block.
We define the protocol more formally in this section. However, we intentionally omit some details, such as those regarding how participants must continually post proofs related to their pledged storage, as they are not relevant to our analysis. Instead, we assume that all participants continuously maintain one unit of storage.
Furthermore, in practice in Filecoin~\cite{filecoin} a participant with power $x$ that is elected twice in the same epoch will create only one block that \emph{weighs} twice more. This is not relevant to our analysis, so we ignore it and prefer a flat model wherein a participant elected twice simply creates two blocks under two different identities. Such a model will favor an adversary as we illustrate
\iffcsubmission
\else{in Section~\ref{sec:attack} }\fi
and hence renders our analysis stronger.

Due to space limitations, a pseudocode representation of the algorithms described in this section can be found in Appendix~\ref{app:pseudocode}.

\subsection{Leader Selection Protocol}

EC's leader election is inspired by Algorand's cryptographic sortition~\cite{gilad2017algorand}.
Briefly, the leader selection relies on a Verifiable Random Function (VRF)~\cite{micali1999verifiable} that takes as input the drand output value for that epoch. 
In each epoch, each participant will compute 
$\vrfproof_\sk(\seed)=(y=G_\sk(\seed),p=\pi_\sk(\seed))$ where $\seed$ is the drand value.
If $y$ is below a predefined value $\target$ that is a parameter of the protocol, then that participant is elected leader. Any other participant can then use $p$ in order to verify that the random value $y$ was computed correctly (i.e., $\vrfverify_\pk(\seed,y, p)=1$) and that the participant is indeed an elected leader.
The value of $\target$ is chosen such that on expectation $m$ leaders are elected in each epoch.
$m$ is a parameter of the EC protocol currently set to $m=5$.

Proving that the leader selection mechanism is secure is outside the scope of this paper, as similar results were already proven in, e.g., Algorand~\cite{gilad2017algorand}. Instead, we assume that in each epoch, there is an independent random number of participants that are elected leaders and that the number of leaders in each epoch follows a Poisson distribution of parameter $m$. For a coalition that consists of a fraction $\alpha$ of all the participants, their number of elected leaders in an epoch follows a Poisson distribution of parameter $\alpha\times m$.

\subsection{Block and Tipset Structure}
A block is composed of a header and a payload. The payload includes transactions as well as other messages necessary for maintaining the set of participants up to date. We omit its content in this analysis.

When a participant is elected to create a block, they include in the \emph{header} of the block their proof of eligibility (i.e., the VRF proof), an epoch number (the epoch at which the block was created), a proof of storage called $\textsf{WinningPost}$ to prove that they maintain the storage they have pledged (we omit the details of such proof) and finally a pointer to a set of \emph{parent} blocks. For a block $\block$, we denote $\block.\parent$ its parents set and $\block.\epoch$ its epoch number.
The parents of a block must satisfy certain conditions.
First, they must all be in the same epoch, and that epoch needs to be smaller than the block's epoch.
Second, all parent blocks need to have the same set of parents themselves.
Each set of blocks that are in the same epoch and have the same set of parents is called a \emph{tipset} and denoted $\tipset$.
Formally, a tipset $\tipset$ is a non-empty set of blocks: $\tipset = \{\block_1,\cdots,\block_r\}$, each of which belongs to the same epoch, i.e., $\block_1.\epoch = \cdots = \block_r.\epoch$ and has the same set of parents, i.e., $\forall (\block_i,\block_j)\in\tipset^2:\;\: \block_i.\parent = \block_j.\parent$.
Since all blocks in a tipset have the same parent, we abuse the notation and denote $\tipset.\parent$ to denote the parent of tipset $\tipset$. Similarly, $\tipset.\epoch$ denotes the tipset epoch.
We note that $\tipset.\parent$ is a tipset itself.

Since each block references a set of blocks, a Directed Acyclic Graph (DAG) structure can be inferred from each block or tipset, where the blocks are the vertices and the references to parents are the edges. 
Similarly, the set of tipsets referencing each other as parents form a \emph{chain}. For example, Figure~\ref{fig:tipset-1} represents a chain of 4 tipsets (including the genesis) and a blockDAG of 9 blocks.
Formally, a chain $\chain$ is then a set of ordered tipsets $\chain = \{\tipset_0,\tipset_1,\cdots,\tipset_l\}$ such that $\tipset_i.\parent = \tipset_{i-1}$ for all $i>1$. By convention, we have $\tipset_1.\parent = \tipset_0 = \{\text{Genesis block}\}$ and $\tipset_i = \emptyset$ if there is no block in epoch $i$. We note $\chain[\tipset_i]= \{\tipset_0,\tipset_1,\cdots,\tipset_i\}$.
Similarly, for a tipset $\tipset$, we can infer the associated chain, denoted $\chain[\tipset]$ as follows: $\chain[\tipset] = \{\tipset_0,\cdots,(\tipset.\parent).\parent,\tipset.\parent,\tipset\} $.

\subsection{Fork Choice Rule and Weight Function}\label{sec:fcr}
In order to decide which tipset to include as its parents, EC provides a \emph{weight function} that assigns a weight to different tipsets. The fork choice rule will then consist of choosing the tipset with the heaviest weight.
In practice, EC's weight function~\cite{filecoin} is a complex function of (1) the number of blocks in the chain and (2) the total amount of storage committed to the chain. Moreover, the total amount of storage is taken far in the past to ensure that everyone agrees on it. 
Since in our analysis we assume a static model where the set of participants
is fixed during the execution of the protocol, we only take into consideration
the number of blocks in the chain. We discuss the impact of this simplification in Section~\ref{sec:limitations-weight}.
Formally, for a tipset $\tipset$, we denote its weight $w(\tipset)$ and have:

$$w(\tipset) = \sum_{\tipset_i\in\chain[\tipset] }|\tipset_i|.$$
In the case of a tie between two chains, a deterministic tie-breaker is used. In practice, the tipset that contains the smallest VRF value is chosen. However, in our analysis we consider a powerful adversary that has the power to decide on ties. 
See Figure~\ref{fig:tipset-2} for a visual representation of two tipsets with equal weight.



\subsection{Mining Algorithm}
We describe the mining algorithm that miners in Filecoin run continuously.
At each epoch $i>0$ each participant with key pair $(\textsf{sk},\textsf{pk})$ performs the following:
\begin{enumerate}
    \item Fetch the drand value for epoch $i$ and verify eligibility by checking $$G_\textsf{sk}(\text{drand}_i)\stackrel{?}{\le} \textsf{target},$$
    where $\textsf{target}$ is chosen such that on expectation $m$ leaders are elected (with $m=5$ in the current implementation).
    \item If elected leader, create a block as follows:
    \begin{itemize}
        \item Choose the tipset with the highest weight (i.e., the most blocks) and reference it as the block's parent.
        \item Include a proof of eligibility (i.e., the VRF value: $\vrfproof_\sk(\drand_i)= (y,p)$), a $\textsf{WinningPost}$ to prove storage maintenance, as well as the payload.
    \end{itemize}
    \item Broadcast the block newly created.
\end{enumerate}
In parallel, whenever they receive a new block in epoch $i$, participants verify its validity and, if it is valid, add it to their blockDAG.
A block is valid if and only if:
\begin{enumerate}
    \item The election proof $(y,p)$ is valid, i.e., $\vrfverify_\pk(\drand_i,y,p) =1$ and $y\le \target$.
    \item The $\textsf{WinningPost}$ is valid (details omitted).
    \item All the transactions in the payload are valid (details omitted).
    \item All its parent blocks form a valid tipset, i.e.:
    \begin{itemize}
        \item They all belong to the same epoch.
        \item They all have the same parents.
        \item They are all valid blocks.
    \end{itemize}
\end{enumerate}

We analyze the backbone of EC in a static setting and hence omit some details of the protocol. For example, in practice, the leader election mechanism uses a \emph{lookback parameter}, meaning that only participants who pledged their storage sufficiently in the past are eligible for block creation. Because we consider a flat and static model, these details are not relevant to our analysis.

\section{Security Definitions}

\noindent{\bf Security properties.} We consider the standard security properties  of \emph{robust transaction ledgers} defined for blockchain systems~\cite{backbone,dembo2020everything}. 
We start by defining a transaction ledger and confirmed transactions in the ledger.

\begin{definition}[Transaction ledger generated by a chain $\C$]
Given a chain $\C$, a transaction ledger $\mathcal{L}$ generated by $C$ is a deterministic, totally-ordered and append-only list of transactions. In particular, if $\C_1$ is a prefix of $\C_2$, then $\mathcal{L}_1$ generated by $\C_1$ is a prefix of $\mathcal{L}_2$ generated by $\C_2$.  
\end{definition}
For example, one way to generate a transaction ledger from a chain $\chain$ is to order the transactions from $\chain$ by order of chronological appearance (i.e., epoch number where they appeared in the chain) and lexicographical order. Any deterministic rule is however valid and we leave this unspecified.

\begin{definition}[Confirmed transaction parameterized by $\tau\in\mathbb{R}$]
\label{def:conf}
If a transaction {\sf tx} in the ledger appears in a block which is mined in epoch $j\le i-\tau$, then {\sf tx} is said to be $\tau$-\emph{confirmed} in epoch $i$.
\end{definition}



Our goal is to generate a transaction ledger that satisfies {\em persistence}  and {\em liveness} as defined in~\cite{backbone,dembo2020everything}.
Together, persistence and liveness guarantee a robust transaction ledger; transactions will be adopted to the ledger and be immutable.  

\begin{definition}[Robust transaction ledger from \cite{backbone,dembo2020everything}]
    \label{def:public_ledger}
    A blockchain protocol $\Pi$ maintains a robust transaction ledger if the generated ledger satisfies the following two properties:
    \begin{itemize}
        \item (Persistence) Parameterized by $\tau \in \mathbb{R}$. If a transaction {\sf tx} becomes  $\tau$-confirmed at epoch $i$ in the view of one honest node, then {\sf tx} will be at least $\tau$-confirmed in the same position in the ledger by all honest nodes for every epoch $k\ge i$.
        \item (Liveness) Parameterized by $u \in \mathbb{R}$, if a transaction {\sf tx} is received by all honest nodes at epoch $i$, then after epoch $i+u$ all honest nodes will contain {\sf tx} in the same place in the ledger forever.
    \end{itemize}
\end{definition}

\noindent {\bf Notations.} We then define random variables and stochastic processes of interest and their properties.

Let $\alpha$ and $\beta$ be the collective fraction of storage power controlled by honest nodes and malicious nodes, respectively ($\alpha + \beta = 1$). We follow the notations of~\cite{bagaria2019prism}.
Let $H[r]$ and $Z[r]$ be the number of blocks mined by the honest nodes and by the malicious nodes in epoch $r$, then $H[r]$, $Z[r]$ are independent Poisson random variables with means $(1-\beta)m$ and $\beta m$ respectively~\cite{filecoin} (the value of the $\textsf{target}$ parameter is chosen to ensure this).
In addition, the random variables $\{H[0], H[1], \cdots\}$ and $\{Z[0], Z[1],\cdots\}$ are independent of one another, since the value provided by drand to feed the leader election is random.
We now define the auxiliary random variables $X[r]$ and $Y[r]$ as follows. If at epoch $r$ an honest node mines at least one block (i.e., $H[r] \geq 1$), then $X[r] = 1$ and epoch $r$ is called a \emph{successful} epoch, otherwise $X[r]=0$. 
If at epoch $r$ honest nodes mine exactly one block (i.e., $H[r] = 1$), then $Y[r] = 1$ and epoch $r$ is called a \emph{uniquely successful} epoch, otherwise $Y[r] = 0$. Epoch $r$ is called an \emph{isolated successful} epoch if it further satisfies that there is no honest block in epoch $r-1$ (i.e., $H[r-1]=0$ and $Y[r]=1$).
Further, $X[r_1,r_2]$ and $Y[r_1,r_2]$ are the number of successful and uniquely successful epochs, respectively, in the interval $(r_1,r_2]$, and $H[r_1,r_2]$ and $Z[r_1,r_2]$ are the number of blocks mined by honest nodes and by the adversary respectively in the interval $(r_1,r_2]$. 

In EC, chains may have equal weights. For simplicity and generality, we assume tie-breaking always favors the adversary. This means that the persistence will be broken as long as there are two sufficiently long forks with equal weights.

Given a chain of tipsets $\C$, let $\C[r]$ be the chain truncated up to blocks in epoch $r$. Further, let $w(\C)$ be the weight of $\C$. Let $W_{\max}[r]$ and $W_{\min}[r]$ be the maximum and minimum weights of chains adopted by honest nodes at the end of epoch $r$. Then, by our network model, we have:
\begin{equation}
    W_{\min}[r] \leq W_{\max}[r] \leq W_{\min}[r+1].
\end{equation}
Even if some honest nodes' chains are ``behind'' in epoch $r$, by our re-broadcast mechanism, their view for epoch $r$ will catch up with the rest of the honest nodes in epoch $r+1$. Furthermore, honest participants always extend the heaviest chain they are aware of, hence the inequality above.

We also have the following minimum honest chain growth property, which is essential to our proof. For $t \geq r+1$,
\begin{equation}
\label{eqn:growth}
    W_{\min}[t] \geq W_{\min}[r+1] + X[r+1,t] \geq  W_{\max}[r] + X[r+1,t].
\end{equation}

This inequality again follows from the fact that honest participants always extend the heaviest chain they know of. However, it could be that different honest participants have different views and thus create blocks on different chains, hence why we consider $X$ in the inequality above and not $H$.

\section{Security Proof}\label{sec:proof}
In this section, we prove our main theorem, Theorem~\ref{thm:ec} stated below, parameterized by the security parameter $\kappa$.
The proof will proceed in multiple steps.
We extend the technique of Nakamoto blocks developed in \cite{dembo2020everything}. We first define the notion of Nakamoto epochs in EC and prove that the honest blocks mined in Nakamoto epochs remain in the heaviest chain forever. Then we show that Nakamoto epochs exist and appear frequently regardless of the adversarial strategy. Straightforwardly, the protocol satisfies liveness and persistence: transactions can enter the ledger frequently through the Nakamoto epochs, and once they enter, they remain at a fixed location in the ledger.

\begin{theorem}
    \label{thm:ec}
    If $\beta m < 1-e^{-(1-\beta)m}$, then EC generates a robust transaction ledger that satisfies {\em persistence} (parameterized by $\tau=\kappa$) and {\em liveness} (parameterized by $u=\kappa$) in Definition~\ref{def:public_ledger} with probability at least $1-e^{-\Omega(\kappa^{1-\epsilon})}$, for any $0 < \epsilon <1$.
\end{theorem}

\subsection{Nakamoto epochs}

Let us define the events: 
\begin{equation*}
    E_{rs} =\{ \mbox{event that $Z[r-1,t]  < X[r+1,t]$ for all $t \geq s$}\},
\end{equation*}

\begin{equation*}
    F_s = \bigcap_{0 \leq r \leq s-2} E_{rs},
\end{equation*}

\begin{equation*}
    U_s =\{ \mbox{event that epoch $s$ is an isolated successful epoch}\} = \{H[s-1]=0, Y[s] = 1\},
\end{equation*}
and
\begin{equation*}
    G_s = F_s \cap U_s.
\end{equation*}

We will call epoch $s$ a \emph{Nakamoto epoch} if the event $G_s$ occurs. And we have the following lemma.

\begin{lemma}
\label{lem:converge}
If epoch $s$ is a Nakamoto epoch, then the unique honest block mined in epoch $s$ is contained in any future chain $\C[t]$, $t \geq s$.
\end{lemma}

\begin{proof}
Let $b_s$ be the unique honest block mined in epoch $s$. We will argue by contradiction. Suppose $G_s$ occurs and let $t \geq s$ be the smallest $t$ such that $b_s$ is not contained in $\C[t]$, an honest chain adopted by some honest node at the end of epoch $t$. 
Let $b_r$, mined in epoch $r$, be the last honest block on $\C[t]$ (which must exist, because the genesis block is by definition honest). If $r > s$, then $\C[r-1]$ is the prefix of $\C[t]$ before block $b_r$, and does not contain $b_s$ (because $\C[r-1]$ is a prefix of $\C[t]$) contradicting the minimality of $t$. So $b_r$ must be mined before or in epoch $s$. Since epoch $s$ is an isolated successful epoch, we further know that $r \leq s-2$. The part of $\C[t]$ after block $b_r$ must consist of all malicious blocks by the definition of $b_r$. Note that this may also include malicious blocks in epoch $r$ (i.e., \textit{headstart} of the adversary). Hence, we have an upper bound for the weight of $\C[t]$.
\begin{equation}
    w(\C[t]) \leq W_{\max}[r] + Z[r-1,t] < W_{\max}[r] + X[r+1,t],
\end{equation}
where the first inequality is illustrated in Figure~\ref{fig:bridge}, and the second inequality follows from the fact that event $F_s$ occurs. We also have a trivial lower bound: $w(\C[t]) \geq W_{\min}[t]$. Therefore, we have
\begin{equation}
    W_{\min}[t] <  W_{\max}[r] + X[r+1,t],
\end{equation}
which contradicts the minimum honest chain growth property (Eqn.~\ref{eqn:growth}).
\end{proof}

\begin{figure}
    \centering
    \includegraphics[width=0.85\textwidth]{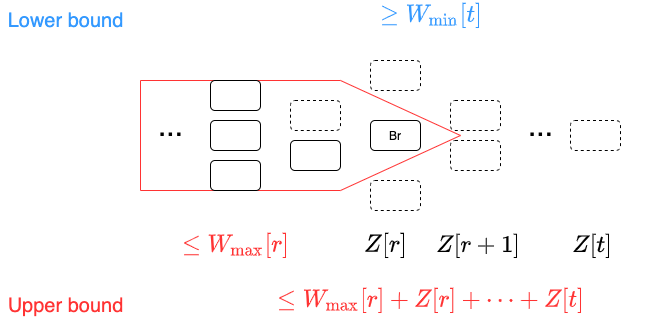}
    \caption{Upper bound and lower bound of the weight of $\C[t]$ in the proof of Lemma~\ref{lem:converge}. Blocks with dotted lines are adversarial blocks. The parent links are omitted for readability; each block has all blocks from the previous epoch as parents.}
    \label{fig:bridge}
\end{figure}

Note that Lemma~\ref{lem:converge} implies that if $G_s$ occurs, then the entire chain leading to the unique honest block mined in epoch $s$ from the genesis is stabilized after epoch $s$.

\subsection{Occurrence of Nakamoto epochs}
Although the existence of Nakamoto epochs ensures that the block at this epoch will be finalized, i.e., it will appear in every honest future chain, the question now remains whether Nakamoto epochs exist at all and, if so, at what frequency they appear. We start answering this question by proving in the next lemma that Nakamoto epochs have a strictly positive probability of happening, i.e., the probability of each epoch being a Nakamoto epoch is strictly positive.

\begin{lemma}
\label{lem:infinite_many_G}
If  $\beta m < 1-e^{-(1-\beta)m}$, then there exists $p > 0$ such that  $P(G_s) \ge p$ for all $s$.
\end{lemma}

\begin{proof}
Let $1- e^{-(1-\beta)m} = (1+\varepsilon) \beta m$ for some $\varepsilon >0$. The random processes of interest start from epoch $0$. To look at the system in stationarity, let us extend them to $-\infty < r < \infty$. This extension allows us to extend the definition of $E_{rs}$ to all $r,s$, $-\infty < r< s < \infty$, and define $\hat{F}_s$ and $\hat{G}_s$ to be:
$$ \hat{F}_s = \bigcap_{r \leq s-2} E_{rs}$$ and 
$$ \hat{G}_s = \hat{F}_s \cap U_s.$$

Note that $\hat{G}_s \subset G_s$, so to prove that $P(G_s) \ge p$ for all $s$, it suffices to prove that $P(\hat{G}_s) \ge p$ for all $s$. This is proved next.

Let $\mathcal{W}_r = (H[r],Z[r])$. Let's note that the variables $\{\mathcal{W}_r\}$'s are independent from each other, hence $\hat{F}_s \cap U_s = \bigcap_{r \leq s-2} (E_{rs} \cap U_s)$ has a time-invariant dependence on $\{\mathcal{W}_r\}$, which means that $P(\hat{G}_r)$ does not depend on $r$. Then it suffices to show $P(\hat{G}_0) > 0$.

\begin{align*}
    P(\hat{G}_0) &= P(\hat{F}_0|U_0)P(U_0 )   \\
    &=  P(\hat{F}_0|U_0)P(H[-1]=0)P(H[0]=1) \\
    &= (1-\beta)m e^{-2(1-\beta)m} P(\hat{F}_0|U_0).
\end{align*}

Then, it remains to show that $P(\hat{F}_0|U_0)>0$.

Recall that
\begin{equation*}
    \hat{F}_0 = \mbox{event that $Z[r-1,t]  < X[r+1,t]$ for all $t \geq 0$ and $r \leq -2$.}
\end{equation*}
Let 
\begin{equation*}
    \hat{B}_{rt} = \mbox{event that $Z[r-1,t]  \geq X[r+1,t]$},
\end{equation*}
then
\begin{equation*}
    \hat{F}_0^c = \bigcup_{t \geq 0, r \leq -2}  \hat{B}_{rt}.
\end{equation*}
Let us fix a particular integer $n > 2$, and define:
\begin{equation*}
    M_n = \mbox{event that $Z[-n-1,n]=0$}.
\end{equation*}
Then 
\begin{eqnarray}
P(\hat{F}_0 | U_0) & \ge & P(\hat{F}_0|U_0,M_n)P(M_n|U_0) \nonumber\\
& = & \left ( 1 - P(\cup_{t \geq 0,r \leq -2} \hat{B}_{rt}|U_0,M_n) \right) P(M_n|U_0) \nonumber\\
& \ge & \left ( 1 - \sum_{t \geq 0,r \leq -2} P(\hat{B}_{rt}|U_0,M_n) \right) P(M_n|U_0) \nonumber\\
& \ge &  ( 1 - a_n - b_n - c_n) P(M_n|U_0), \label{eqn:up_bound}
\end{eqnarray}
where
\begin{eqnarray*}
a_n & := & \sum_{(r,t): -n \leq r \leq -2 \text{~and~}  0 \leq t \le n} P(\hat{B}_{rt}|U_0,M_n),\\
b_n & := & \sum_{(r,t):  r \leq -2 \text{~and~} t > n }P(\hat{B}_{rt}|U_0,M_n), \\
c_n & := & \sum_{(r,t):  r < -n \text{~and~} t \geq 0 }P(\hat{B}_{rt}|U_0,M_n).
\end{eqnarray*}

Note that the cases where $t>n$ and $r<-n$ are counted twice in $b_n$ and $c_n$, but this is fine because we only need a lower bound. Next we will bound $P(\hat{B}_{rt}|U_0, M_n)$. Consider three cases:

\noindent {\bf  Case 1:} $-n \leq r \leq -2 \text{~and~}  0 \leq t \le n$:
\begin{eqnarray*}
&&P(\hat{B}_{rt}|U_0, M_n) \\
&=& P(Z[r-1,t]  \geq X[r+1,t] | H[0] = 1, H[-1] = 0, Z[-n-1,n]=0) \\
&=& P(X[r+1,t] \leq 0 | H[0] = 1, H[-1] = 0, Z[-n-1,n]=0) \\
&\leq& P(X[0] \leq 0 | H[0] = 1, H[-1] = 0, Z[-n-1,n]=0) \\
&=& 0.
\end{eqnarray*}
Summing these terms, we have $a_n = 0$.

\noindent {\bf Case 2:} $r \leq -2 \text{~and~} t > n$:

\begin{eqnarray*}
&&P(\hat{B}_{rt}|U_0, M_n) \\
&=& P(Z[r-1,t]  \geq X[r+1,t] | H[0] = 1, H[-1] = 0, M_n) \\
&\leq& P((Z[r-1,-2] + Z[n,t]) \geq (X[r+1,-2] + X[0,t] + 1)) \\
&\leq& P((Z[r-1,-2] + Z[n,t]) \geq (X[r+1,-2] + X[0,t])) \text{//(*)} \\
&\leq& P(Z[r+3,t] \geq X[r+3,t]) \text{~~~~~~~//$n>1$ and $Z$'s and $X$'s are i.i.d r.v.s. }\\
&\leq& P(Z[r+3,t] \geq X[r+3,t],X[r+3,t] \geq (1-\varepsilon/4)(1- e^{-(1-\beta)m})(t-r-3)) \\
&~&~~ + P(Z[r+3,t] \geq X[r+3,t], X[r+3,t]  < (1-\varepsilon/4)(1- e^{-(1-\beta)m})(t-r-3)) \\
&\leq& P(Z[r+3,t] \geq (1+\varepsilon/4)\beta m(t-r-3)) \\
&~&~~ + P(X[r+3,t] < (1-\varepsilon/4)(1- e^{-(1-\beta)m})(t-r-3)) \\
&\leq& A_0 e^{-\alpha_0 \varepsilon^2(t-r-3)}
\end{eqnarray*}
for some positive constants $A_0, \alpha_0$ independent of $n,r,t$. 
Inequality (*) follows from counting $t-r-3$ epochs on the right hand side of the equation and $\leq t-r-3$ epochs on the left hand side.
The last two inequalities follow from Lemma~\ref{lem:poisson}, Lemma~\ref{lem:chernoff} and the facts that $1- e^{-(1-\beta)m} = (1+\varepsilon) \beta m$ and $\frac{1+\varepsilon/4}{1-\varepsilon/4} < 1+\varepsilon$.
Summing these terms, we have:
\begin{eqnarray*}
b_n & = & \sum_{(r,t):  r \leq -2 \text{~and~} t > n}P(\hat{B}_{rt}|U_0,M_n) \\
& \leq  & \sum_{ (r,t): r \leq -2 \text{~and~} t > n} \left [A_0 e^{-\alpha_0 \varepsilon^2(t-r-3)} \right] := \bar{b}_n,
\end{eqnarray*}
which is bounded and moreover $\bar{b}_n \rightarrow 0$ as $n \rightarrow \infty$.

\noindent {\bf Case 3:} $r < -n \text{~and~} t \geq 0$:

\begin{eqnarray*}
&&P(\hat{B}_{rt}|U_0, M_n) \\
&=& P(Z[r-1,t]  \geq X[r+1,t] | H[0] = 1, H[-1] = 0, M_n) \\
&\leq& P((Z[r-1,-n-1] + Z[2,t]) \geq (X[r+1,-2] + X[0,t] + 1)) \\
&\leq& P((Z[r-1,-n-1] + Z[2,t]) \geq (X[r+1,-2] + X[0,t])) \\
&\leq& P(Z[r+3,t] \geq X[r+3,t]) \text{~~~~~~~//$n>2$ and $Z$'s and $X$'s are i.i.d r.v.s.}\\
&\leq& P(Z[r+3,t] \geq X[r+3,t],X[r+3,t]  \geq (1-\varepsilon/4)(1- e^{-(1-\beta)m})(t-r-3)) \\
&~&~~ + P(Z[r+3,t] \geq X[r+3,t], X[r+3,t] < (1-\varepsilon/4)(1- e^{-(1-\beta)m})(t-r-3)) \\
&\leq& P(Z[r+3,t] \geq (1+\varepsilon/4)\beta m(t-r-3)) \\
&~&~~ + P(X[r+3,t] < (1-\varepsilon/4)(1- e^{-(1-\beta)m})(t-r-3)) \\
&\leq& A_0 e^{-\alpha_0 \varepsilon^2(t-r-3)}
\end{eqnarray*}
for some positive constants $A_0, \alpha_0$ independent of $n,r,t$. The last two inequalities follow from Lemma~\ref{lem:poisson}, Lemma~\ref{lem:chernoff} and the facts that $1- e^{-(1-\beta)m} = (1+\varepsilon) \beta m$ and $\frac{1+\varepsilon/4}{1-\varepsilon/4} < 1+\varepsilon$.
Summing these terms, we have:
\begin{eqnarray*}
c_n & = & \sum_{(r,t):  r < -n \text{~and~} t \geq 0}P(\hat{B}_{rt}|U_0,M_n) \\
& \leq  & \sum_{ (r,t): r < -n \text{~and~} t \geq 0} \left [A_0 e^{-\alpha_0 \varepsilon^2(t-r-3)} \right] := \bar{c}_n,
\end{eqnarray*}
which is bounded and moreover $\bar{c}_n \rightarrow 0$ as $n \rightarrow \infty$.

Substituting these bounds in (\ref{eqn:up_bound}) we finally get:
\begin{equation*}
    P(\hat{F}_0|U_0) > (1- \bar{b}_n -\bar{c}_n)P(M_n|U_0).
\end{equation*}
By setting $n$ sufficiently large such that $\bar{b}_n$ and $\bar{c}_n$ are sufficiently small, we conclude that $P(\hat{G}_0)> 0$.

\end{proof}

\subsection{Waiting time for Nakamoto epochs}

We have established the fact that the event $G_s$ has  $P(G_s) \ge p > 0$ for all $s$. But how long do we need to wait for such an epoch to occur? We answer this question in the following lemma, wherein we provide a bound on the probability that in a interval $(j,j+k]$ of $k$ consecutive epochs, there are no Nakamoto epochs, i.e.,  a bound on:
$$q(j,j+k] :=P(\bigcap_{s=j+1}^{j+k} G_s^c),$$
where  $G_s^c$ is the complement of $G_s$. 
\begin{lemma}
\label{lem:round}
If  $\beta m < 1-e^{-(1-\beta)m}$, then there exist constants $\alpha, A$ so that for all $j,k\geq 0$,
\begin{equation}
q(j,j+k] \leq  A \exp(- \alpha \sqrt{k}).
\end{equation}
\end{lemma}

\begin{proof}
Following the definition in Lemma~\ref{lem:infinite_many_G}, let
\begin{equation*}
    \hat{B}_{rt} = \mbox{event that $Z[r-1,t]  \geq X[r+1,t]$}.
\end{equation*}
Similar to the calculation in Lemma~\ref{lem:infinite_many_G}, we have 
\begin{equation}
    P(\hat{B}_{rt}) \leq A_1 e^{-\alpha_1 \varepsilon^2(t-r)}
    \label{eqn:single}
\end{equation}
for some positive constants $A_1, \alpha_1$ independent of $r,t$.

Also we have 
\begin{equation}
    G_s^c = F_s^c \cup U_s^c = \bigcup_{(r,t): r \leq s-2, t \geq s} \hat{B}_{rt} \cup U_s^c.
\end{equation}
Divide $(j,j+k]$ into $\sqrt{k}$ sub-intervals of length $\sqrt{k}$ (assuming $\sqrt{k}$ is a integer), so that the $i$-th sub-interval is:
$$\mathcal{J}_i : = [j+1 +  (i-1) \sqrt{k}, j+ i\sqrt{k}].$$

Now look at the first, fourth, seventh, etc sub-intervals, i.e. all the $i = 1 \mod 3$ sub-intervals. Introduce the event that in the $\ell$-th ($1 \mod 3$) sub-interval ($\mathcal{J}_{3\ell+1}$), a pure adversarial chain that is rooted at a honest block (or more accurately a tipset including at least one honest block) mined in that sub-interval ($\mathcal{J}_{3\ell+1}$) or in the previous ($0 \mod 3$) sub-interval  ($\mathcal{J}_{3\ell}$) catches up with a honest block in that sub-interval ($\mathcal{J}_{3\ell+1}$) or in the next ($2 \mod 3$) sub-interval  ($\mathcal{J}_{3\ell+2}$). 

Formally,
$$C_{\ell}=\bigcap_{s \in \mathcal{J}_{3\ell+1}}
\bigcup_{(r,t): r \in \mathcal{J}_{3\ell} \cup \mathcal{J}_{3\ell+1}, r \leq s-2, t \geq s, t \in \mathcal{J}_{3\ell+1} \cup \mathcal{J}_{3\ell+2} } \hat{B}_{rt} \cup U_s^c.$$
Note that for distinct $\ell$, the events $C_\ell$'s  are independent since $\hat{B}_{rt}$'s in different $C_\ell$'s do not have overlap (the $\mathcal{J}$ intervals were cut specifically for this purpose). Also, we have
\begin{equation}
    \label{eqn:short_range}
    P(C_{\ell})\leq P(\mbox{no Nakamoto epoch in $\mathcal{J}_{3\ell+1}$}) = 1-p < 1
    \end{equation}
by Lemma~\ref{lem:infinite_many_G}. 

Introduce the atypical events:
\begin{eqnarray*}
    B &=& \bigcup_{(r,t): r \in (j,j+k] \mbox{~or~} t \in (j,j+k], r <t, t - r \geq  \sqrt{k}} \hat{B}_{rt} \;, \text{ and }\\
    \tilde{B} &=& 
    \bigcup_{(r,t):r\leq j,j+k<t} \hat{B}_{rt}\;.
\end{eqnarray*}
The events $B$ and $\tilde{B}$ are the events that an adversarial chain catches up with an honest block far ahead (more than $\sqrt{k}$ epochs).

By (\ref{eqn:single}) and an union bound we have that 

\begin{eqnarray*}
&~& P(B) \\
&\leq& \sum_{(r,t): r \in [j+1,j+k] \mbox{~or~} t \in [j+1,j+k], r < t, t - r \geq \sqrt{k}} A_1 e^{-\alpha_1 \varepsilon^2(t-r)} \\
&\leq& \sum_{r=j+1}^{j+k} \big( \sum_{t=r+\sqrt{k}}^{\infty} A_1 e^{-\alpha_1 \varepsilon^2(t-r)} \big)  + \sum_{t=j+1}^{j+k} \big( \sum_{r=0}^{t-\sqrt{k}} A_1 e^{-\alpha_1 \varepsilon^2(t-r)} \big) \\
&\leq& 2k\frac{A_1 e^{-\alpha_1 \varepsilon^2 \sqrt{k}}} {1-e^{-\alpha_1 \varepsilon^2}},
\end{eqnarray*}
and
\begin{eqnarray*}
P(\tilde{B}) &\leq& \sum_{(r,t):r\leq j,t>j+k} A_1 e^{-\alpha_1 \varepsilon^2(t-r)} \\
&\leq& \sum_{r=0}^{j} \big( \sum_{t=j+k+1}^{\infty} A_1 e^{-\alpha_1 \varepsilon^2(t-r)} \big)\\
&=& \sum_{r=0}^{j} \frac{A_1 e^{-\alpha_1 \varepsilon^2 (j+k+1-r)}}{1-e^{-\alpha_1 \varepsilon^2}} \\
&\leq& \frac{A_1 e^{-\alpha_1 \varepsilon^2 (k+1)}}{(1-e^{-\alpha_1 \varepsilon^2})^2}.
\end{eqnarray*}

Now, we have: 
\begin{eqnarray}
\label{eqn:short_long}
&&q(j,j+k] \nonumber\\
&\leq& P(\mbox{no Nakamoto epoch in $\bigcup_{\ell=0}^{\sqrt{k}/3} \mathcal{J}_{3\ell+1}$}) \nonumber \\
&\leq& P(\mbox{no isolated successful epoch in $\bigcup_{\ell=0}^{\sqrt{k}/3} \mathcal{J}_{3\ell+1}$}) + P(B) + P(\tilde{B}) + P(\bigcap_{\ell=0}^{\sqrt{k}/3} C_{\ell}) \nonumber\\
&=& e^{-\Omega(k)} + P(B)+P(\tilde{B}) + (P(C_{\ell}))^{\sqrt{k}/3} \label{eqn:ind}\\
&\leq& e^{-\Omega(k)} + 2k\frac{A_1 e^{-\alpha_1 \varepsilon^2 \sqrt{k}}} {1-e^{-\alpha_1 \varepsilon^2}} + \frac{A_1 e^{-\alpha_1 \varepsilon^2 (k+1)}}{(1-e^{-\alpha_1 \varepsilon^2})^2} + (P(C_{\ell}))^{\sqrt{k}/3} \nonumber \\
&\leq&  A \exp(- \alpha \sqrt{k}) \label{eqn:union}
\end{eqnarray}
for some positive constants $A$ and $\alpha$. The equality~\eqref{eqn:ind} is due to the independence of $C_\ell$'s and the inequality~\eqref{eqn:union} is due to \eqref{eqn:short_range}. Hence the lemma follows.
\end{proof}

We can also tighten the exponent, but at the cost of larger constants in the bound. The proof of the following lemma is almost verbatim identical with the proof of Lemma~\ref{lem:round}, and its detailed explanation can be found in Appendix~\ref{app:prooftightexponent}.

\begin{lemma}
\label{lem:round1}
If  $\beta m < 1-e^{-(1-\beta)m}$, then there exist constants $\alpha_\epsilon, A_\epsilon$ so that for all $j,k\geq 0$,
\begin{equation}
q(j,j+k] \leq  A_\epsilon \exp(- \alpha_\epsilon k^{1-\epsilon}),
\end{equation}
for any $0 < \epsilon <1$.
\end{lemma}

\subsection{Persistence and liveness}

Equipped with all the previous lemmas, we can now prove the persistence and liveness properties of EC for $\beta m < 1-e^{-(1-\beta)m}$.
\begin{proof}[Proof of Theorem~\ref{thm:ec}]
Suppose current epoch is $r$. Then by Lemma~\ref{lem:round1}, with probability at least $1-e^{-\Omega(\kappa^{1-\epsilon})}$, there is at least one Nakamoto epoch in the interval $(r-\kappa,r]$. Let epoch $s \in (r-\kappa,r]$ be a Nakamoto epoch. Then by Lemma~\ref{lem:converge}, the chain up to epoch $s-1$ is permanent since the unique honest block in epoch $s$ never leaves the heaviest chain. Hence EC is persistent with probability at least $1-e^{-\Omega(\kappa^{1-\epsilon})}$. The liveness of EC is simply a consequence of the frequent occurrence of Nakamoto epochs. Particularly, for each honest transaction, either it will be included by an honest block $B$ in a Nakamoto epoch or  it has already been included by $B$'s ancestors.
\end{proof}

\section{n-split Attack}
\label{sec:attack}

In order to confirm whether an adversary with power $\beta m \ge 1-e^{-(1-\beta)m}$ can indeed break the persistence and liveness properties of the system, we consider the following $n$-split attack. 
\subsection{Attack description}
The attacker tries to split the honest participants among $n$ chains such that in each epoch, at most one honest block is added to each chain (i.e., no two honest players mine on the same chain). To do this, the attacker creates $n$ copies of one of its block (each copy has the same election proof, but different payloads) and sends one different block to each of the $n$ honest players; see illustration in Figure~\ref{fig:epoch-boundary-step1}. 
To maintain the split for a long period, the adversary must repeat the attack at every epoch where 
at least one honest block is mined.
In this case, the weight of the chain of each honest player will increase by two: one honest block and one adversarial block. For example in Figure~\ref{fig:epoch-boundary-step1}, since blocks C and D are both honest (i.e., created by honest miners), by the next epoch, epoch 3, all the participants will have received them and use the deterministic tie breaker to all decide to mine on the same tipset, say \{D\}. 
Hence the adversary must create equivocating blocks in epoch 2 as well in order to ensure that in epoch 3, honest miners all choose different tipsets to append their block to. 
In Figure~\ref{fig:ep-boundary-step2}, the adversary sends equivocating blocks $E_1$ and $E_2$ to prevent blocks $H$ and $F$ from being appended to the same chain.
These figures include only two honest blocks at epoch 2 and 3 for clarity. In practice, the adversary will create as many equivocating blocks as there are honest miners to ensure that everyone sees a different block and that no two honest participants mine on the same tipset.

Whenever there is no honest block mined in an epoch, the attacker does nothing. In this case, the weight of the chain of each honest player will not increase. Meanwhile, the attacker also reuses all its blocks to build a private chain, i.e., a chain that it does not broadcast to other participants and that does not include any honest blocks. The expected chain growth of the adversary's private chain is $\beta m$.
The weight of the honest chain increases by two if there is at least one honest block mined in an epoch (which happens with probability $1-e^{-(1-\beta)m}$), and 0 otherwise (which happens with probability $e^{-(1-\beta)m}$). Hence, the expected chain growth of the honest chain is $2(1-e^{-(1-\beta)m})$. Therefore, this attack succeeds with non-negligible probability when $\beta m > 2(1-e^{-(1-\beta)m})$, i.e., when the adversarial chain grows at a higher rate than the honest split chains. Rather than specifying the exact success probability, we demonstrate that it remains constant, independent of the confirmation depth $\tau$, as defined in Definition~\ref{def:conf}. Let $L$ be adversarial lead, i.e., the adversary has a lead of $L$ additional private blocks over the public heaviest chain. Suppose that with probability $p_{L_0}$, the initial lead before the attack starts is $L_0$.  Although $p_{L_0}$ decreases with $L_0$, it remains non-zero because there's a chance the adversary could mine $L_0$ blocks before the honest nodes mine any block. Note that in the $n$-split attack, as long as the adversarial lead $L > 0$, the adversary can invariably split the honest nodes across $n$ chains. Therefore, the adversarial private chain will grow faster than the honest public chain when $\beta m > 2(1-e^{-(1-\beta)m})$. According to the standard random walk (with drift) theory~\cite{feller1971introduction}, $L$ goes to 0 only with a probability of $e^{-O(L_0)}$. This implies that the $n$-split attack succeed with probability at least $p_{L_0}(1-e^{-O(L_0)})$, for any value of the confirmation depth $\tau$.

\begin{figure}
\centering
\begin{subfigure}[t]{.3\textwidth}
  \centering
    \includegraphics[width=.6\textwidth]{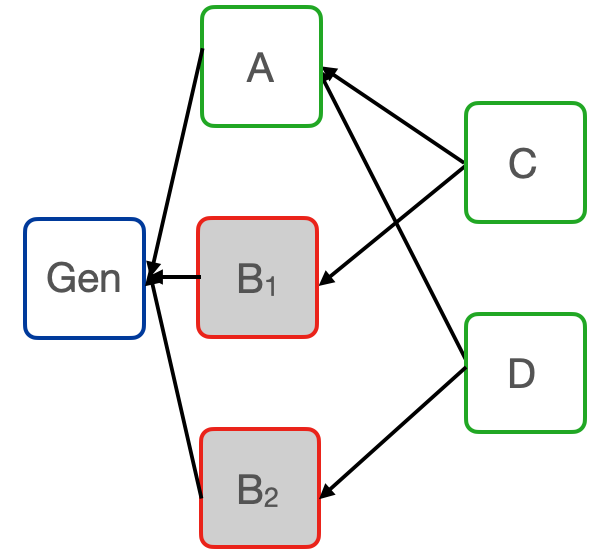}
    \caption{The adversary sends two different blocks $B_1$, $B_2$ in epoch 1 such that in epoch 2, honest blocks $C$, $D$  have different parents and hence cannot be included in a tipset. The honest power is thus split between different tipsets' chains.}
    \label{fig:epoch-boundary-step1}
\end{subfigure}%
\hspace{1em}%
\begin{subfigure}[t]{.3\textwidth}
  \centering
    \includegraphics[width=.8\textwidth]{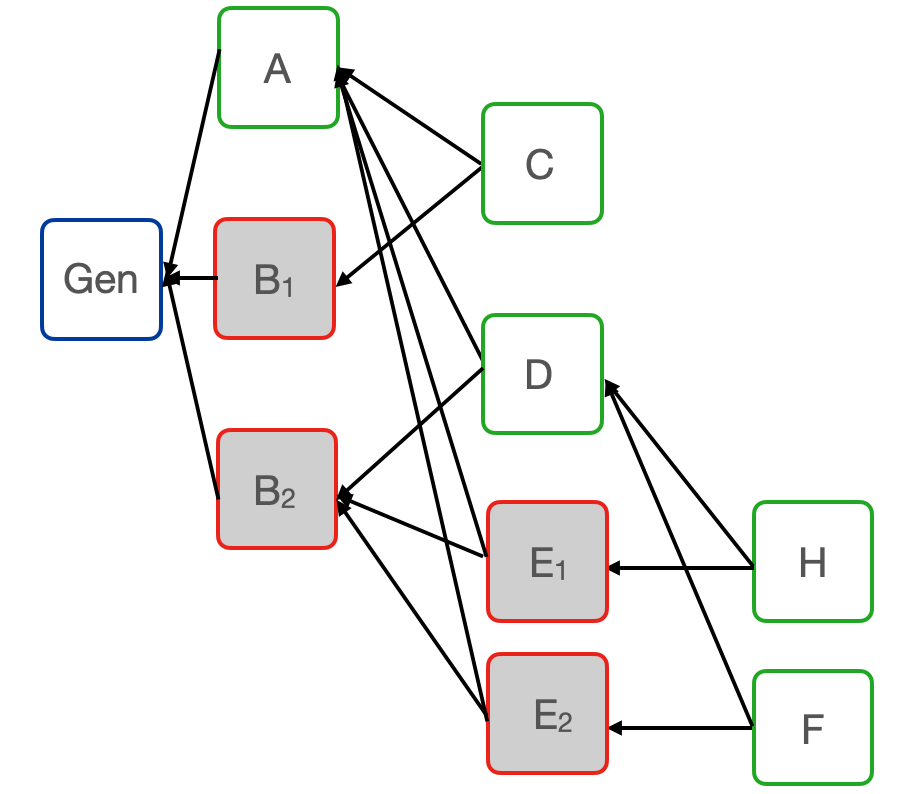}
    \caption{The adversary keeps the network split in epoch 3 by creating two equivocating blocks: $E_1$ and $E_2$ in epoch 2. In epoch 3, honest blocks $H$ and $F$ are mined on two different tipsets. }
  \label{fig:ep-boundary-step2}
\end{subfigure}
\begin{subfigure}[t]{.3\textwidth}
  \centering
    \includegraphics[width=.8\textwidth]{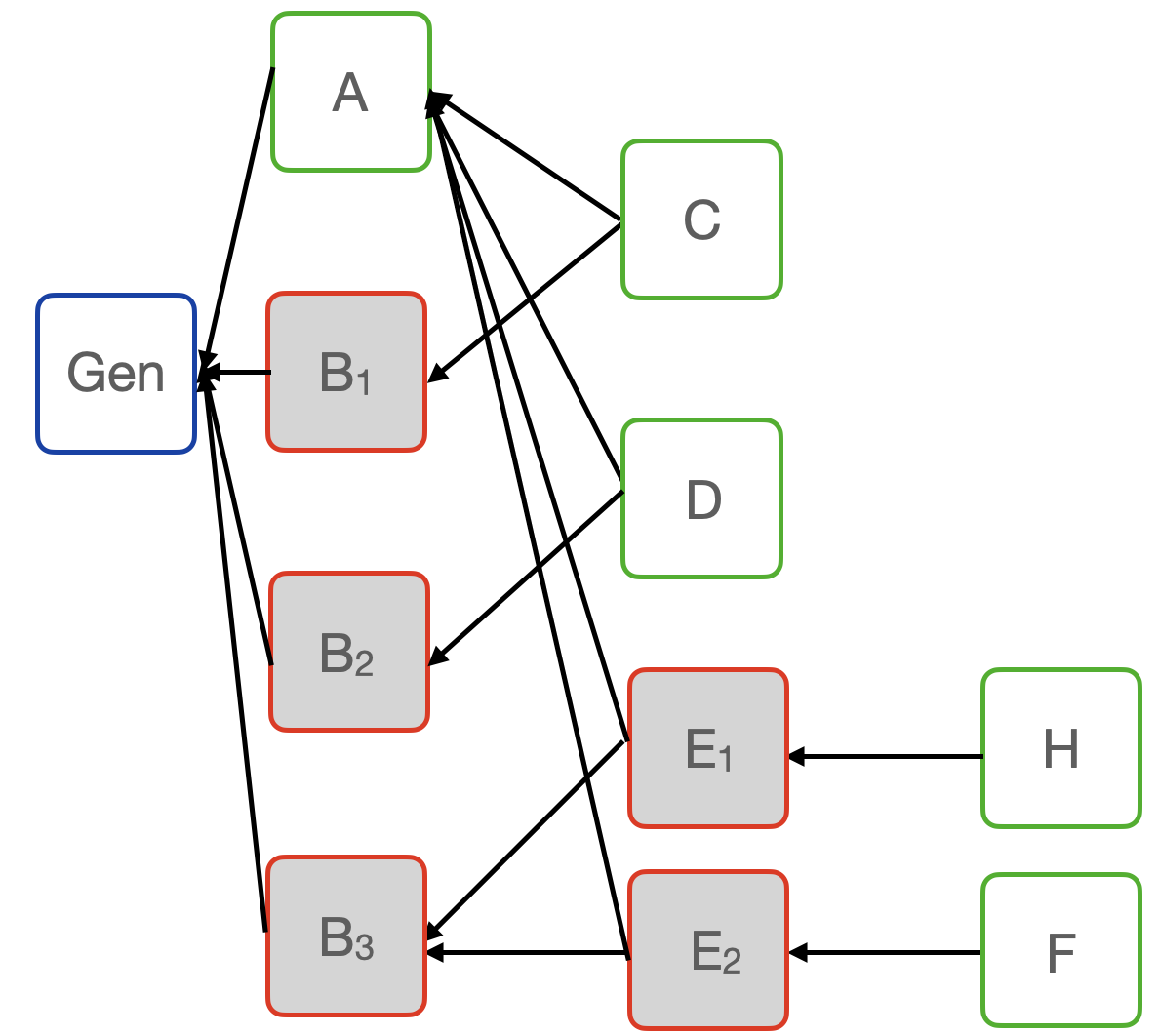}
    \caption{If ties are always broken in favor of the adversary, the adversary can instead create another block $B_3$ in epoch 1. In epoch 2, blocks $E_1$ and $E_2$ are preferred to $C$ and $D$, hence in epoch 2, all the tipsets' chains increase by only one block. The attack is repeated in the next epoch, as long as the adversary has enough blocks to create a fork as heavy as the honest chains.}
  \label{fig:ep-boundary-tie}
\end{subfigure}
\caption{n-split attack. Each block filled in grey is an equivocating block, meaning they were created by the adversary using the same leader election proof in one epoch. Each green block is an honest block.}
\label{fig:epochboundaries}
\end{figure}

Some numerical results: for $m = 3$, we have $\beta>0.512$; for $m = 5$, we have $\beta>0.382$; and $\beta>0.284$ for $m = 7$. Approximately, $\beta \gtrsim 2/m$.
Intuitively, with this attack, the threshold is inversely proportional to $m$, as with a bigger $m$, the adversary is elected leader more often and thus has more opportunities to keep the network split. If the adversary were elected leader on a less regular basis, it would be harder to keep sending equivocating blocks and thus keep the network split for longer.
Note that this threshold applies specifically to the attack described above, and it differs from the security threshold identified in Theorem~\ref{thm:ec}. This is because different thresholds may exist for different attacks.

Following the standard model of longest-chain analysis~\cite{backbone,david2018ouroboros,dembo2020everything}, we give the power of tie-breaking to the attacker (i.e., tie-breaking always favors the attacker's block). Recall that the same assumption is made in our model (Section~\ref{sec:model}) and proof (Section~\ref{sec:proof}).
For example in Figure~\ref{fig:ep-boundary-tie}, if we assume that the adversarial blocks (in red) will always be favored in the case of two chains with the same number of blocks, then the adversary does not need to create a block on the same tipset as honest blocks (as in Figure~\ref{fig:ep-boundary-step2}). 
The adversary will instead create another block $B_3$ and mine yet on another tipset than the honest participant in epoch 2. In epoch 3 the adversarial chains ending in tipsets $\{E_1\}$ and $\{E_2\}$ are preferred over $\{C\}$ or $\{D\}$, hence honest blocks $H$ and $F$ are mined on different forks and each fork's weight increased by only one in epoch 2, as opposed to 2 in Figure~\ref{fig:ep-boundary-step2}.
By repeating this attack at each epoch, the weight of the chain of each honest player will only increase by one when there is at least one honest block mined.
Following the same argument as above, this attack succeeds with non-negligible probability when $\beta m > (1-e^{-(1-\beta)m})$, i.e., when the adversarial chain grows at a higher rate than the honest split chains. We notice that this now matches the security threshold in Theorem~\ref{thm:ec}, hence proving that $\beta m = 1-e^{-(1-\beta)m}$ is the tight threshold of the protocol in our security model as defined in Section~\ref{sec:model}. Indeed, Theorem~\ref{thm:ec} proves that no adversary below this threshold can break the security of the EC, and the $n$-chain split attack just described proves that an adversary above this threshold can indeed break the persistence and liveness of the EC.

Some numerical results: for $m = 3$, we have $\beta\simeq 0.293$; for $m = 5$, we have $\beta\simeq 0.196$; and $\beta\simeq 0.143$ for $m = 7$. Approximately, $\beta \simeq 1/m$. As remarked before, the threshold is inversely proportional to the number of leaders elected as, intuitively, being elected more often gives more opportunities to an adversary.

\subsection{Discussion}\label{sec:attack-discussion}


\noindent {\bf Rationality of the attack.}
We note that this attack is detectable as everyone can see that blocks with the same proof of eligibility but different payload were created. In practice, this behaviour is slashable in Filecoin~\cite{filecoin}. However, in Filecoin an adversary has the ability to spread its storage over multiple identities, i.e., create multiple identities that each possesses one unit of storage.
For example in Filecoin the minimum unit of storage that can be pledged to the chain is 32 GiB. As of August 2023 the total storage pledged to the chain is around 11 EiB~\cite{filfox}, hence an adversary that possesses 20\% of the total power, i.e., 2.2 EiB could potentially ``spread'' its storage over $\frac{2.2\times10^{18}}{32\times10^9}\simeq 10^8$ different identities, that each possesses 32 GiB of storage.
At each epoch, except with extremely small probability, the adversary will have a new ``identity'' elected to create a block (it is very unlikely for a miner with 32GiB of storage out of 11 EiB to be elected twice in a row).
Assuming that each identity gets slashed and removed from the list of participants after equivocation, after performing the attack over 1000 epochs, the adversary will still have $9.9999\cdot 10^{7}$ identities left out of $10^8$ and will only be slashed $\frac{10^3}{10^8}=10^{-5}$ of its total collateral. 
In this analysis we considered orders of magnitude rather than exact numbers, so for simplicity we counted only the ``real power'' and did not account for the ``boosted adjusted power''
that can be gained through the FIL+ program~\cite{filplus}.
Note that it is not possible for any honest miner to know which identities belong to the adversary before the equivocation.
Hence excluding equivocating participants from the protocol is not sufficient to prevent the attack.

Furthermore, we note that for the adversary to be slashed, a special transaction, \emph{a fraud proof} transaction must be submitted on-chain by any participant. An adversary that is able to continually exclude honest blocks as is the case with the $n$-split attack may thus in practice never be slashed as no honest participant will get the opportunity to include the slashing transaction on-chain. This is why even when considering incentives and the slashing mechanism in place, this attack is still rational.\\

\noindent {\bf Network control.}
In this attack, we assumed a powerful adversary that not only has the power to break ties in its favor, but has also full control of the network, as specified in Section~\ref{sec:model-network}.
However we remark that for this attack to work, an adversary only needs limited power over the network. Specifically, the adversary needs to be directly connected to every participants but does not need to control the propagation time between honest participants, as we illustrate now.

In Filecoin honest miners will stop accepting blocks for an epoch after a cutoff time. 
For the split to happen the adversary could send each block $B_1,\dots,B_n$ to each different participants $1,\dots,n$ right before the cutoff, i.e., participant $i$ receives block $B_i$ from the adversary just before the cutoff time, ensuring block $B_i$ is accepted by participant $i$. The adversary would need to know the propagation time between itself and each participant to do so, however this is easy to estimate.
Since participant $i$, receives $B_i$ just before the cutoff time, whatever the propagation delay between $i$ and another honest participant $j$ is, the adversary is guaranteed that $j$ will received $B_i$ from $i$ \emph{after} the cutoff time and hence that any honest miner $j\ne i$ will not accept $B_i$.

Furthermore we note that when participant $j$ receives block $B_i$, after the cutoff, $j$ will detect the equivocation as $j$ already received block $B_j$ from the adversary.
However at that point, $j$ has already created its block that includes $B_j$ as a parent, hence it is too late for $j$ to discard $B_j$ due to equivocation.
The mitigation that we propose in Section~\ref{sec:mitigation-cbroadcast} changes this.

\begin{figure}
\centering
\begin{subfigure}[b]{0.8\textwidth}
  \centering
  \includegraphics[width=\linewidth]{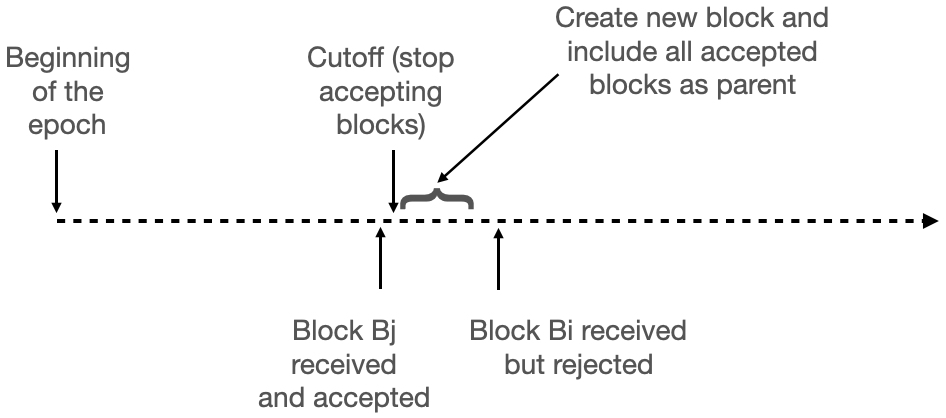}
  \caption{Current case.}
  \label{fig:cutoff1}
\end{subfigure}%
\vspace{\floatsep}
\begin{subfigure}[b]{0.7\textwidth}
  \centering
    \includegraphics[width=1\linewidth]{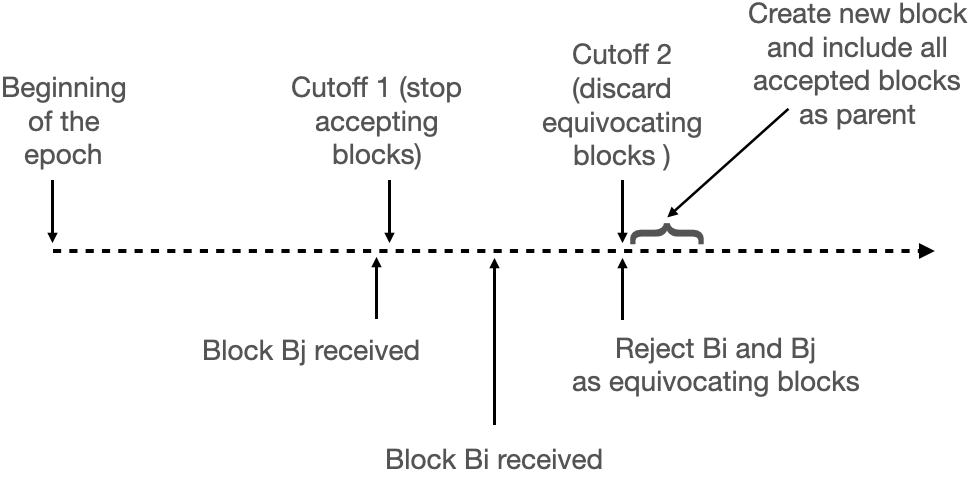}
    \caption{Case with a (simplified) consistent broadcast.}
  \label{fig:cutoff2}
\end{subfigure}
\caption{The dashed arrow represents the arrow of time. The different cutoff and arrival time are marked with vertical arrows. The adversary ensures that $j$ receives $B_j$ right before the cutoff so $j$ accepts $B_j$. In the first case, by the time $j$ sees an equivocation, it is too late as $B_j$ was already included as a parent. In the second case, assuming $j$ receives $B_i$ before the second cutoff, then $j$ will discard $B_j$ and not include it as a parent.}
\label{fig:cutoff}
\end{figure}

\section{Mitigations}\label{sec:mitigation}
We propose two possible mitigations to the $n$-split attack described in Section~\ref{sec:attack} which could also help increase the security threshold of EC.

\subsection{Replace EC by the Longest-chain Protocol in SPC}
One solution to the $n$-split attack is to remove the notion of tipsets and instead change EC to the longest-chain protocol (i.e., Ouroboros family of protocols~\cite{kiayias2017ouroboros,david2018ouroboros,badertscher2018ouroboros} in the proof-of-stake setting), where one block has exactly one parent. 
In the longest-chain protocol, the effort of an adversary to split the network would have much less impact on the overall security of the protocol since each chain can increase by one block at each epoch at most anyway.
Furthermore, moving to the longest-chain protocol allows for inheritance of all the security properties (e.g., a security threshold of 50\%) of all proof-of-stake protocols based on that setting~\cite{david2018ouroboros,bagaria2019proof,dembo2020everything}. Dembo et al.~\cite{dembo2020everything} indeed showed that in the longest-chain case, the worst attack is the private attack. Hence, the $n$-split attack described in Section~\ref{sec:attack}, or its variant, would not be the worst attack anymore. However, transitioning from EC to a longest-chain protocol is not a straightforward task. First, it's crucial to understand that merely setting $m = 1$ does not transform EC into a longest-chain protocol, as more than one block can still be added per epoch. Furthermore, our analysis indicates that an EC with $m = 1$ has a security threshold of approximately 43.2\%, as opposed to 50\% in the longest-chain protocol. Consequently, to enhance the security, the concept of a tipset will need to be eliminated. Implementing such a change, however, would necessitate a hard fork.

\subsection{Consistent Broadcast}\label{sec:mitigation-cbroadcast}
Another solution is to use a form of consistent or reliable broadcast~\cite{bracha1985asynchronous,guerraoui2019scalable}.
This type of broadcast prevents an adversary from equivocating (i.e., creating two blocks with the same leader election proof but different contents). 
The consistent broadcast 
consists in adding a second cutoff to the cutoff discussed in Section~\ref{sec:attack-discussion}.
Specifically, it ensures that after participant $j$ received $B_j$ from the adversary (before the first cutoff), $j$ will wait for a ``second'' cutoff before forming its block and including $B_j$ as a parent. When $j$ receives equivocating block $B_i$ from the adversary, $j$ will detect the equivocation and decide not to include $B_j$ (neither $B_i$), hence the attack is mitigated.
This is illustrated in Figure~\ref{fig:cutoff}.
The epoch in Filecoin is thus split as follows: during the first period, participants will store every valid block received in their ``pending blocks'' set. In the second period, after the first cutoff and before the second cutoff, every new valid block received will be stored in the set of ``rejected blocks'' for that epoch. In the third period, after the second cutoff, participants will compare the set of pending blocks and rejected blocks, if they detect equivocating blocks, they are removed from the pending set. Every block that is left in the pending set will then be included in the tipset.
This implementation as it is simple and backward compatible (i.e., only requires a soft-fork) although it assumes synchronicity.
\iffcsubmission
\else{With consistent broadcast, however, the adversary is still able to split the network in two ways. First it could do so by ensuring that only a fraction of the honest nodes accept its block (i.e., the network will be split between the nodes that accept the adversarial block vs those that do not).  
The second way in which the network could be split is if the adversary is elected more than once, say $\ell$ times in one epoch (here we assume that the adversary controls many different participants). 
Then the adversary could similarly ensure that for each block it is able to create (with different election proofs), only some of the nodes accept it.
The network is then split in $2^\ell$ ways (all the combination of accept vs reject for each of the $\ell$ blocks).}\fi
We hypothesize that, under a non-equivocating adversary, the security threshold of EC with $m=5$ is approximately $40\%$. Intuitively, with this level of power, the adversary creates fewer than two blocks per epoch on average; furthermore, without the ability for equivocation, it is impossible for the adversary to maintain two chains of equal weights over numerous epochs.
We leave a formal proof as future work.

\section{Limitations and Future Work}\label{sec:limitations-weight}
For practical reasons, this work made a few simplifying assumptions. We discuss them here.\\

\noindent {\bf Incentive consideration.}
This work considers the classic model of honest vs malicious participants and does not address the rationality of participants. A formal study of incentive compatibility is also important for understanding the security of EC. However, we leave this for future work.
Assuming a fully malicious adversary that is willing to lose money to attack the system, a scenario we consider in this paper, makes for a stronger proof than assuming a rational adversary.
In 
\iffcsubmission
\else{Section~\ref{sec:attack-discussion}, }\fi
we discussed why it is realistic to consider an irrational adversary for the $n$-split attack we proposed, as slashing may not always be possible if the adversary has the ability to censor transactions.
It still remains to show that the honest strategy is compatible with a rational strategy even in the presence of an adversary. We leave this for future work.\\

\noindent {\bf Weight function.}
In our analysis, we only took into consideration the number of blocks in the chain for the weight function. We leave as future work an analysis that also considers the total storage, as specified in Expected Consensus~\cite{filecoin}. Specifically we believe that a complex weight function allows for more vectors of attack and that an adversary could use this to try to blow the weight of its own chain. For example, the adversary could remove its storage from the main chain and thus decrease the weight of the main chain, while privately creating an alternative chain that would be heavier because it has more storage pledged. 
The mechanisms for maintaining and removing storage are, however, complex and ignored in this work. 
For simplicity, we thus consider the weight of a tipset to simply be equal to the number of blocks referenced in its blockDAG.


\section{Conclusion}
In this paper we presented a formal analysis of Expected Consensus, a sub-protocol of Filecoin's  Storage Power Consensus, and we proposed two concrete ways to improve SPC's security.
One of our mitigations, using consistent broadcast, is currently being implemented as a Filecoin Improvement Proposal. It remains an open problem to quantify the new security threshold of EC with this fix, although our proofs remain valid in this case, hence the security threshold is at least such that $\beta m< 1- e^{-(1-\beta)m}$ as proved in Section~\ref{sec:proof}.
Furthermore, we made many simplifying assumptions in this work. It would be interesting to relax these in the future work; e.g., by extending this proof to the dynamic and asynchronous case, considering the more complex variant of the weight function or incorporating incentives.


\bibliography{ref}

\appendix

\section*{Appendix}

\section{Pseudocode for EC}\label{app:pseudocode}
The main algorithm is presented in Algorithm~\ref{pseudocode:main} and the algorithm for block validation is presented in Algorithm~\ref{pseudocode:block-validation}.
\algdef{SE}[EVENT]{Event}{EndEvent}[1]{\textbf{upon event}\ #1\ \algorithmicdo}{\algorithmicend\ \textbf{event}}%
\algtext*{EndEvent}

\algdef{SE}[PARAM]{Param}{EndParam}[1]{\textbf{Parameters:}  #1}{\algorithmicend\ \textbf{parm2}}%
\algtext*{EndParam}

\algdef{SE}[IMPORT]{Import}{EndImport}[1]{\textbf{import} #1}{\algorithmicend }%
\algtext*{EndImport}

\algdef{SE}[INIT]{Init}{EndInit}[1]{\textbf{Init:} #1}{\algorithmicend }%
\algtext*{EndInit}

\begin{algorithm*}
\begin{algorithmic}[1]
\Import
\State drand
\State ForkChoiceRule
\State Broadcast
\State VRF
\State $\textsf{isValid}$ (Algorithm~\ref{pseudocode:block-validation})
\EndImport
\Param
\State epochLength
\State $m$
\State $\target$ \Comment{Chosen such that $m$ leaders are elected on expectation}
\EndParam
\Init 
\State epochNumber $\gets 0$
\State blockDAG$\gets$\{Genesis Block\}
\EndInit
\Event{time.Now() \% epochLength == 0}
\Comment{Beginning of the epoch}
\State epochNumber $\gets $ epochNumber +1 
\State $\seed\gets$ drand(epochNumber)
\State $(y,p)\gets\vrfproof_\sk(\seed)$
\If{$y\le \target$}
\State $\tipset\gets$ForkChoiceRule(blockDAG) \Comment{Choose the DAG with the most blocks}
\State $\block\gets$ CreateBlock($\tipset,(y,p)$, epochNumber,$\textsf{WinningPost}$ payload)
\State Broadcast($\block$)
\EndIf
\EndEvent
\Event{Receiving block $\block$}
\If{$\textsf{isValid}(\block)==1$}
\State blockDAG.append($\block$)
\EndIf
\EndEvent
\end{algorithmic}
\caption{Main algorithm}\label{pseudocode:main}
\end{algorithm*}

\begin{algorithm*}
\begin{algorithmic}[1]
\State \textbf{Input:} block $\block$
\Import
\State drand
\State $\textsf{isPayloadValid}$
\State $\textsf{isStorageValid}$
\EndImport
\State Parse ($\tipset,(y,p)$, epochNumber,$\textsf{WinningPost}$ payload) $\gets\block$
\State $\seed\gets$ drand(epochNumber)
\If{$\vrfverify_\pk(\seed,y,p) ==0$ or $y> \target$}\Comment{Check the election proof}
\State return 0
\EndIf
\If{$\textsf{isPayloadValid}$(payload)==0}\Comment{Check the payload}
\State return 0
\EndIf
\If{$\textsf{isStorageValid(WinningPost)}==0$}\Comment{Check the storage proof}
\State return 0
\EndIf
\For{$\block_i\in\tipset$}\Comment{Check validity of parent blocks}
\If{$\textsf{isValid}(\block_i)==0$}
\State return 0
\EndIf
\EndFor
\State return 1
\end{algorithmic}
\caption{$\textsf{isValid}(\block)$}\label{pseudocode:block-validation}
\end{algorithm*}

\section{Concentration Inequalities}
\label{app:ineq}

\begin{lemma}[Chernoff]
\label{lem:chernoff}
Let $X = \sum_{i=1}^{n} X_i$, where $X_i = 1$ with probability $p_i$ and $X_i = 0$ with probability $1 - p_i$, and all $X_i$'s are independent. Let $\mu = \mathbb{E}[X] = \sum_{i=1}^{n} p_i$. Then for $0 < \delta <1$, $\mathbb{P} (X > (1+\delta)\mu) < e^{-\Omega(\delta^2 \mu)}$ and $\mathbb{P} (X < (1-\delta)\mu) < e^{-\Omega(\delta^2 \mu)}$.

\end{lemma}

\begin{lemma}[Poisson]
\label{lem:poisson}
Let $X$ be a Poisson random variable with rate $\mu$. Then for $0 < \delta <1$, $\mathbb{P} (X > (1+\delta)\mu) < e^{-\Omega(\delta^2 \mu)}$ and $\mathbb{P} (X < (1-\delta)\mu) < e^{-\Omega(\delta^2 \mu)}$.

\end{lemma}

\section{Proof of Lemma~\ref{lem:round1}}\label{app:prooftightexponent}


\begin{proof}
By Lemma~\ref{lem:round}, we know that the statement is true for $\epsilon = 1/2$ (and hence for $\epsilon > 1/2$). Now we prove for any $\epsilon >0$ by recursively applying the bootstrapping procedure in Lemma~\ref{lem:round}. We use induction to show that the statement is true for all $\epsilon \geq m_n$, where $m_n = \frac{1}{n+1}, n\geq 1$. As a base step, this is true for $n = 1$. Now we assume it holds for $n \geq 1$ and prove for the step $n+1$.

Divide $(j,j+k]$ into $k^{\frac{1}{n+2}}$ sub-intervals of length $k^{\frac{n+1}{n+2}}$ (assuming $k^{\frac{1}{n+2}}$ is a integer), so that the $i$-th sub-interval is:
$$\mathcal{J}_i : = [j+1 +  (i-1) k^{\frac{n+1}{n+2}}, j+ i k^{\frac{n+1}{n+2}}].$$

Now look at the first, fourth, seventh, etc sub-intervals, i.e. all the $i = 1 \mod 3$ sub-intervals. Introduce the event that in the $\ell$-th ($1 \mod 3$) sub-interval ($\mathcal{J}_{3\ell+1}$), an pure adversarial chain that is rooted at a honest block (or more accurately a tipset including at least one honest block) mined in that sub-interval ($\mathcal{J}_{3\ell+1}$) or in the previous ($0 \mod 3$) sub-interval  ($\mathcal{J}_{3\ell}$) catches up with a honest block in that sub-interval ($\mathcal{J}_{3\ell+1}$) or in the next ($2 \mod 3$) sub-interval  ($\mathcal{J}_{3\ell+2}$). 
Formally,
$$C_{\ell}=\bigcap_{s \in \mathcal{J}_{3\ell+1}}
\bigcup_{(r,t): r \in \mathcal{J}_{3\ell} \cup \mathcal{J}_{3\ell+1}, r \leq s-2, t \geq s, t \in \mathcal{J}_{3\ell+1} \cup \mathcal{J}_{3\ell+2} } \hat{B}_{rt} \cup U_s^c.$$
Note that for distinct $\ell$, the events $C_\ell$'s  are independent since $\hat{B}_{rt}$'s in different $C_\ell$'s do not have overlap. Also by the previous induction step, we have
\begin{eqnarray}
    \label{eqn:short_range1}
    P(C_{\ell}) &\leq& P(\mbox{no Nakamoto epoch in $\mathcal{J}_{3\ell+1}$}) \nonumber\\
    &\leq& A_{m_n} \exp(-\alpha_{m_n} {(k^{\frac{n+1}{n+2}})}^{1-m_n}) \nonumber\\
    &=& A_{m_n} \exp(-\alpha_{m_n} k^{\frac{n}{n+2}}).
\end{eqnarray}

Introduce the atypical events:
\begin{eqnarray*}
    B &=& \bigcup_{(r,t): r \in (j,j+k] \mbox{~or~} t \in (j,j+k], r <t, t - r \geq  k^{\frac{n+1}{n+2}}} \hat{B}_{rt} \;, \text{ and }\\
    \tilde{B} &=& 
    \bigcup_{(r,t):r\leq j,j+k<t} \hat{B}_{rt}\;.
\end{eqnarray*}
The events $B$ and $\tilde{B}$ are the events that an adversarial chain catches up with an honest block far ahead (more than $k^{\frac{n+1}{n+2}}$ epochs). Following the calculations in Lemma~\ref{lem:round}, we have
\begin{eqnarray*}
&~&P(B) \\
&\leq& \sum_{(r,t): r \in [j+1,j+k] \mbox{~or~} t \in [j+1,j+k], r < t, t - r \geq k^{\frac{n+1}{n+2}}} A_1 e^{-\alpha_1 \varepsilon^2(t-r)} \\
&\leq& \sum_{r=j+1}^{j+k} \big( \sum_{t=r+k^{\frac{n+1}{n+2}}}^{\infty} A_1 e^{-\alpha_1 \varepsilon^2(t-r)} \big)  + \sum_{t=j+1}^{j+k} \big( \sum_{r=0}^{t-k^{\frac{n+1}{n+2}}} A_1 e^{-\alpha_1 \varepsilon^2(t-r)} \big) \\
&\leq& 2k\frac{A_1 e^{-\alpha_1 \varepsilon^2 k^{\frac{n+1}{n+2}}}} {1-e^{-\alpha_1 \varepsilon^2}},
\end{eqnarray*}
and
\begin{eqnarray*}
P(\tilde{B}) &\leq& \sum_{(r,t):r\leq j,t>j+k} A_1 e^{-\alpha_1 \varepsilon^2(t-r)} \\
&\leq& \sum_{r=0}^{j} \big( \sum_{t=j+k+1}^{\infty} A_1 e^{-\alpha_1 \varepsilon^2(t-r)} \big)\\
&=& \sum_{r=0}^{j} \frac{A_1 e^{-\alpha_1 \varepsilon^2 (j+k+1-r)}}{1-e^{-\alpha_1 \varepsilon^2}} \\
&\leq& \frac{A_1 e^{-\alpha_1 \varepsilon^2 (k+1)}}{(1-e^{-\alpha_1 \varepsilon^2})^2}.
\end{eqnarray*}

Now, we have
\begin{eqnarray}
\label{eqn:short_long1}
&&q(j,j+k] \nonumber\\
&\leq& P(\mbox{no Nakamoto epoch in $\bigcup_{\ell=0}^{k^{\frac{1}{n+2}}/3} \mathcal{J}_{3\ell+1}$}) \nonumber \\
&\leq& P(\mbox{no isolated successful epoch in $\bigcup_{\ell=0}^{k^{\frac{1}{n+2}}/3} \mathcal{J}_{3\ell+1}$}) + P(B) + P(\tilde{B}) + P(\bigcap_{\ell=0}^{k^{\frac{1}{n+2}}/3} C_{\ell}) \nonumber\\
&=& e^{-\Omega(k)} + P(B)+P(\tilde{B}) + (P(C_{\ell}))^{k^{\frac{1}{n+2}}/3} \label{eqn:ind1}\\
&\leq& e^{-\Omega(k)} + 2k\frac{A_1 e^{-\alpha_1 \varepsilon^2 k^{\frac{n+1}{n+2}}}} {1-e^{-\alpha_1 \varepsilon^2}} + \frac{A_1 e^{-\alpha_1 \varepsilon^2 (k+1)}}{(1-e^{-\alpha_1 \varepsilon^2})^2} + (P(C_{\ell}))^{k^{\frac{1}{n+2}}/3} \nonumber \\
&\leq&  A_{m_{n+1}} \exp(- \alpha_{m_{n+1}} k^{\frac{n+1}{n+2}}) \label{eqn:union1} \\
&=& A_{m_{n+1}} \exp(- \alpha_{m_{n+1}} k^{1-m_{n+1}}) \nonumber
\end{eqnarray}
for some positive constants $A_{m_{n+1}}$ and $\alpha_{m_{n+1}}$. Equality~\eqref{eqn:ind1} is due to the independence of $C_\ell$'s and inequality~\eqref{eqn:union1} is due to \eqref{eqn:short_range1}.

So we know that the statement holds for all $\epsilon \geq m_{n+1} = \frac{1}{n+2} $. And since $\lim_{n\rightarrow\infty} m_n = 0$, this concludes the lemma.

\end{proof}

\end{document}